\documentclass[]{revtex4-2}
\usepackage{iceberg}

\algnewcommand\algorithmicforeach{\textbf{for each}}
\algdef{S}[FOR]{ForEach}[1]{\algorithmicforeach\ #1\ \algorithmicdo}

\usepackage{dsfont}

\pgfkeys{tikz/mymatrix/.style={matrix of math nodes,inner sep=0pt,row sep=0em,column sep=0em,nodes={inner sep=6pt}}}

\newcommand{\supportlogical}{\mathfrak{L}}

\newcommand{\fv}{\mathfrak{v}}
\renewcommand{\norm}[1]{\left \| #1 \right \|}

\begin{document}

\title{Fast surgery for quantum LDPC codes}

\author{Nou\'edyn Baspin}
\affiliation{Iceberg Quantum, Sydney}

\author{Lucas Berent}
\affiliation{Iceberg Quantum, Sydney}

\author{Lawrence Z.\ Cohen}
\affiliation{Iceberg Quantum, Sydney}

\author{{\scriptsize\texttt{\{nouedyn,lucas,larry\}@iceberg-quantum.com}}}

\begin{abstract}
    Quantum LDPC codes promise significant reductions in physical qubit overhead compared with topological codes. 
    However, many existing constructions for performing logical operations come with distance-dependent temporal overheads. 
    We introduce a scheme for performing generalized surgery on quantum LDPC codes using a constant number of rounds of syndrome measurement. 
    The merged code in our scheme is constructed by taking the total complex of the base code and a suitably chosen homomorphic chain complex. 
    We demonstrate the applicability of our scheme on an example multi-cycle code and assess the performance under a phenomenological noise model, showing that fast surgery performs comparably to standard generalized surgery with multiple rounds. 
    Our results pave the way towards fault-tolerant quantum computing with LDPC codes with both low spatial and temporal overheads.
\end{abstract}

\maketitle

\section{Introduction}
Quantum error correction will be an essential component of any utility-scale quantum computer. 
This requirement has led to daunting numbers for the space and time required to perform useful quantum algorithms. 
For example, the predominant approach towards fault-tolerance has focused on two-dimensional topological codes~\cite{Dennis2002}, such as the surface code. 
Fault-tolerant architectures based on surface codes have been heavily optimized over the past two decades~\cite{Litinski2019, gidney2025factor2048bitrsa}. Nevertheless, the number of physical qubits needed to perform useful algorithms is orders of magnitude greater than what is available today. 
In recent years significant attention has been applied to rectify this problem by considering quantum LDPC codes~\cite{Breuckmann2021}. 
These codes utilize non-local connectivity to significantly reduce the physical overhead required for quantum error correction. 
Such non-locality is a feature of many of the most promising platforms~\cite{Yoneda2021, Malinowski2023, Bluvstein2022, bombin2021}. 
It is now known that `good' quantum LDPC codes, such that both the number of encoded information and the distance increase linearly with the size of the code, exist~\cite{Panteleev2022, Leverrier2022}. Furthermore, several instances of small QLDPC codes whose physical overheads significantly outperform topological codes for the same logical failure rates have been constructed~\cite{Bravyi2024}. 
Thus, quantum LDPC codes are a promising route not just for future large scale systems, but also for early fault-tolerant systems.

To execute fault-tolerant algorithms it is not enough to just encode the logical information in error correcting codes. 
We also require methods for reliably performing logical gates on the encoded information. 
Initial research into quantum LDPC codes focused on their performance as memories~\cite{Tremblay2022, Grospellier2021}, as it was unclear how to perform universal processing directly on the information. 
Instead, information was `moved' into other codes, such as the surface code, for performing the bulk of the processing~\cite{Breuckmann2017}. 
However, recently there has been a significant amount of work on performing logical operations directly on the encoded information in quantum LDPC codes. 
This approach has the advantage of not requiring topological codes and thus maintaining the small physical overheads that make quantum LDPC codes advantageous. 
Various methods have been developed to perform logical gates directly in LDPC codes. 
These include methods that generalize surface code lattice surgery~\cite{Cohen2022, Williamson2024, Ide2025}, as well as methods that utilize various combinations of transversal gates and code automorphisms~\cite{Quintavalle2023, Huang2023, Xu2025}.

While quantum LDPC codes have been shown to be exceptionally promising for reducing the physical overhead of fault-tolerant quantum computing, many schemes for performing logical gates still exhibit large time overheads, similar to those encountered in schemes based on topological codes~\cite{Horsman2012}.
In particular, schemes utilizing generalized surgery~\cite{Cohen2022} incur a time overhead that grows linearly with the distance of the code.
For hardware platforms with slower clock cycles, such as trapped ions~\cite{Cai2023} and neutral atoms~\cite{Wintersperger2023}, this time overhead can prove prohibitive to executing large scale fault-tolerant algorithms.
It is hence imperative to develop schemes for performing fault-tolerant gates quickly, ideally in a time that is constant regardless of the size of the system.
Recently there have been several directions taken to address this problem, including utilizing algorithmic fault-tolerance in conjunction with transversal gates to increase the logical clock speed of surface code computing~\cite{zhou2025}.
Such ideas have also been explored in the context of quantum LDPC codes~\cite{malcolm2025}.
One drawback of these approaches however is that they require using specific code families.
This is in contrast to generalized surgery, which allows for the execution of Clifford measurements on arbitrary error correcting codes~\cite{he2025}.

In this paper we describe a scheme for performing fast generalized surgery on quantum LDPC codes.
We prove that the fault distance of our scheme is lower bounded by the code distance, even when only a constant number (one) of rounds of syndrome extraction are performed.
Our work builds on prior work on generalized surgery~\cite{cowtan2025parallel}, and on previous schemes that utilize techniques from homological algebra~\cite{Huang2023}. Some elements of our framework also relate to mapping cones surgery \cite{Ide2025}. Our results allow us to generalize schemes for performing lattice surgery on codes such as the 3D surface code~\cite{Hillmann2024}.

The paper is laid out as follows.
In Section~\ref{sec:prelims} we give a brief overview of the necessary homological algebra concepts, and discuss prior schemes of generalized surgery and other related work.
In Section~\ref{sec:surgery} we give our main theoretical construction and prove properties about the fault-distance of the resulting procedure given certain assumptions.
In Section~\ref{sec:constructing} we give an explicit, though not necessarily optimal, method for constructing schemes that satisfy the assumptions of Section~\ref{sec:surgery}.
Finally, in Section~\ref{sec:case_study} we work through an example of our scheme on the recently introduced Abelian multi-cycle codes.

\section{Preliminaries} \label{sec:prelims}

\subsection{Chain complexes and homology}

Throughout this paper  we use the toolkit of homological algebra, and so we begin with a brief overview on $\mathbb{F}_2$-homology. 
A $D$-dimensional chain complex, $(C_\bullet, \partial_\bullet)$, over $\mathbb{F}_2$ is a collection of $\mathbb{F}_2$-vector spaces, $\left\{ C_i \right\}_{i=0}^{D}$, and linear boundary maps $\left\{ \partial_i: C_{i+1} \to C_i \right\}_{i=1}^{D}$
\begin{equation}
    C_\bullet =
    \begin{tikzcd}
        C_D \arrow[r, "\partial_D"] & C_{D-1} \arrow[r, "\partial_{D-1}"] & \cdots \arrow[r, "\partial_2"] & C_1 \arrow[r, "\partial_1"] & C_0
    \end{tikzcd}\, ,
\end{equation}
such that the boundary maps satisfy $\partial_i \circ \partial_{i-1} = 0$. 
We assume that each vector space $C_i$ comes with a distinguished basis and its elements are called $i$-chains.
Moreover, the elements of $B_i = \im \partial_{i+1}$ are called $i$-boundaries, and the elements of $Z_i = \ker \partial_i$ are called $i$-cycles. 
The homology groups are defined as $H_i\left(C_\bullet\right)$, i.e., they correspond to cycles that are not boundaries
\begin{equation}
    H_i\left(C_\bullet\right) = \ker \partial_i / \im \partial_{i+1}\, .
\end{equation}
We can similarly define the cohomology groups, which for chain complexes over $\mathbb{F}_2$-vector spaces are defined as
\begin{equation}
    H^i\left(C^{\bullet}\right) = \ker \delta_i / \im \delta_{i-1}\, ,
\end{equation}
where $\delta_i = \partial_{i+1}^T$ for the considered $\mathbb{F}_2$ vector spaces. 
Since the vector spaces are equipped with distinguished bases, we define a norm $\left \| \cdot \right \|$ on the $i$-chains given by the Hamming weight. 
Then the $i$-systole and $i$-cosystole are defined by
\begin{align}
    \text{Syst}_i\left(C_{\bullet}\right) &= \min \left\{ \left\| v \right\| \mid v \in Z_i \setminus B_i \right\}, \\
    \text{Syst}^i\left(C^{\bullet}\right) &= \min \left\{ \left\| v \right\| \mid v \in Z^i \setminus B^i\right\}.
\end{align}
\paragraph*{Remark:} In the conventional literature, both $H_i$ and $H^i$ are often treated as abstract groups. 
In the context of QEC codes however, their representation on $C_i$ is of singular importance, as it carries important data regarding the geometry of the logical operators. 
In the rest of this work, we will slightly misuse notation and identify any (co)homology group $H$ with a basis $\{\vec{u}\}_u \subset Z_i$ that is isomorphic to $H$.

\subsubsection*{Chain maps and total complex}
There exists a natural notion of a morphism for chain complexes. 
Suppose we have two chain complexes, $\left(C_\bullet, \partial_C\right)$ and $\left(D_\bullet, \partial_D\right)$, along with a set of homomorphisms $\gamma_i: D_i \to C_i$. 
If the homomorphisms and the boundary maps satisfy the commutativity condition
\begin{equation}
    \gamma \partial = \partial \gamma \, ,
\end{equation}
then the map $\gamma_i$ describes a chain map between the complexes $C_\bullet$ and $D_\bullet$. 
We can represent this visually via the following commuting diagram
\begin{equation}
    \begin{tikzcd}
        \cdots \arrow[r] & C_{i+1} \arrow[r] & C_{i} \arrow[r] & C_{i-1} \arrow[r] & \cdots \\
        \cdots \arrow[r] & D_{i+1} \arrow[u, "\gamma_{i+1}"] \arrow[r] & D_{i} \arrow[u, "\gamma_{i}"] \arrow[r] & D_{i-1} \arrow[u, "\gamma_{i-1}"] \arrow[r] & \cdots
    \end{tikzcd} \, .
\end{equation}
We can define a new chain complex from the above complex, which is referred to as the total complex and is given by
\begin{equation}
    \tot E_\ell = \bigoplus_{i+j = \ell} E_{i,j}, \quad \partial^E = \sum_{i+j = \ell} \partial^h_i + \partial^v_j \,
\end{equation}
where the horizontal components $\partial^h$ correspond to the boundary maps $\partial_C, \partial_D$, and the vertical components $\partial^v$ to the maps $\gamma_j$.

\subsubsection*{Systolic expansion}
There is a notion of expansion in chain complexes which generalizes the notion of edge expansion in graphs, which can be viewed as $1$-dimensional chain complexes.
\begin{definition} \label{def:systolic_expansion}
    Let $C_\bullet$ be a $D$-dimensional complex and let $\epsilon > 0$. 
    For $i < D$ we say that the vector space $C_i$ is $\epsilon$-systolic expanding if
    \begin{equation}
        \min \left\{ \frac{\left\|\partial_i f\right\|}{\min_{z \in Z_i} \left\|f+z\right\|} \mid f \in C_i \setminus Z_i \right\} \geq \epsilon
    \end{equation}
\end{definition}
There are several explicit constructions of chain complexes that are also systolic expanders~\cite{Evra2016, Oppenheim2025,dinur_cubical_24}.

\subsection{CSS codes as chain complexes}

We denote the Pauli group on $n$ qubits by $\mathcal{P}_n$. 
An $\llbracket n, k, d \rrbracket$ quantum stabilizer code is a $2^k$ dimensional subspace of the $2^n$ dimensional group $\mathcal{P}_n$. 
It is described by an Abelian group $\mathcal{S} \in \mathcal{P}_n$ such that $-I \notin \mathcal{S}$. 
The group $\mathcal{S}$ are the stabilizers of the code. 
We choose a generating set $S$ of $\mathcal{S}$, which are the stabilizer generators of the code. 

An $\llbracket n, k, d \rrbracket$ quantum CSS code is a stabilizer code such that the stabilizer generators are either products of $X$-Pauli operators, or products of $Z$-Pauli operators only. 
A CSS code can thus be specified by two parity check matrices, $H_X \in \mathbb{F}_2^{m_x \times n}$ and $H_Z \in \mathbb{F}_2^{m_z \times n}$. 
The rows of $H_X$ and $H_Z$ correspond to $X$-type stabilizer generators and $Z$-type stabilizer generators, respectively.
Note that the commutativity of the stabilizers enforces the condition $H_X \cdot H_Z^T = 0$. 
Moreover, a code is called a low-density parity-check (LDPC) code if the matrices $H_X,H_Z$ have constant weight, independent of the code size $n$.

A CSS code can be equivalently described by a 2-dimensional chain complex
\begin{equation} \label{eq:css_complex}
    C_\bullet = 
    \begin{tikzcd}[cells={nodes={minimum height=2em}}]
        \mathbb{F}^{m_x}_2 \arrow[r,"H_X^T"] & \mathbb{F}^n_2 \arrow[r, "H_Z"] & \mathbb{F}^{m_z}_2
    \end{tikzcd}\, .
\end{equation}
Nilpotency of the boundary maps is equivalent to the commutativity condition on the parity check matrices. 
The logical Pauli operators are defined as the set of Pauli operators that commute with all the stabilizers but are not stabilizers themselves. 
Hence, they correspond to elements in the first homology and cohomology group
\begin{align}
    L_X &= H_1\left(C_\bullet\right) = \ker H_Z / \im H_X^T \\
    L_Z &= H^1\left(C_\bullet\right) = \ker H_X / \im H_Z^T\, .
\end{align}
The code distance is the minimum Hamming weight of an operator in $L_X \cup L_Z$. 
This can be expressed by the systole and cosystole of the chain complex
\begin{align}
    d_X &= \text{Syst}_1\left(C_\bullet\right) \\
    d_Z &= \text{Syst}^1\left(C^\bullet\right).
\end{align}
Some codes also possess meta-checks which can be used to protect against measurement errors when performing syndrome readout~\cite{Campbell2019}.
A set of $X$-type meta-checks is described by a matrix $M_X \in \mathbb{F}^{m'_x \times m_x}_2$ such that $M_X \cdot H_X = 0$. 
In this case we can describe the code by the $3$-dimensional chain complex
\begin{equation}
    C_\bullet = 
    \begin{tikzcd}[cells={nodes={minimum height=2em}}]
        \mathbb{F}^{m'_x}_2 \arrow[r, "M_X^T"] &
        \mathbb{F}^{m_x}_2 \arrow[r,"H_X^T"] & \mathbb{F}^n_2 \arrow[r, "H_Z"] & \mathbb{F}^{m_z}_2
    \end{tikzcd} \, .
\end{equation}

\subsection{Generalized surgery}
\label{subsec:generalized-surgery}
Generalized surgery is a scheme for performing fault-tolerant computation by measuring logical Pauli operators using ancillary systems.
By attaching the ancillary system to the code block the code can be deformed into a new code, similar to surface code lattice surgery.

It is possible to implement the full Clifford group on $n$ logical qubits using only logical Pauli measurements and ancilla logical qubits~\cite{Litinski2019}. 
Hence, when supplemented with magic states, generalized surgery allows us to perform universal fault-tolerant computation. 
Several schemes have been proposed for performing generalized surgery, however, they all broadly follow a similar pattern, which is encapsulated by the following definition of a lattice surgery ancilla system.
\begin{definition}[Lattice surgery ancilla system]
\label{def:surgery-patch}
    Let $L \in \cP_X^n$ be an $X$-type logical of a stabilizer code $\cC$ with support $\supportlogical \subset [n]$. 
    Then a lattice surgery ancilla system is described by matrices $A_X$, $A_Z$ and maps $\gamma_1$, $\gamma_2$ such that the deformed parity check matrices
    \begin{align}
        H'_X &= \begin{pmatrix}
            H_X & 0 \\
            \gamma_1 & A_X
        \end{pmatrix}, \\
        H'_Z &= \begin{pmatrix}
            H_Z & \gamma_1 \\
            0 & A_Z
        \end{pmatrix},
    \end{align} 
    satisfy the following requirements
    \begin{enumerate}
        \item There exists $\vec{v} \in \rowspace(H_X')$ such that $\supp(\vec{v}) = \supportlogical$;
        \item $\supp(L')\neq \supp(\vec{u})$ $\forall L' \neq L$ and all vectors $\vec{u} \in \rowspace(H'_X)$;
        \item The stabilizer generators commute, i.e., $H'_X{H'}_Z^T = 0$.
    \end{enumerate}
\end{definition}
The first requirement ensures that the logical Pauli operator we wish to measure is in the stabilizer group of the deformed code and hence its value can be inferred via syndrome readouts. 
The second requirement ensures that no other logical Pauli operators are measured during the deformation.

There are of course other properties that are desirable, for example we generally want to ensure that the distance of the deformed code is sufficiently large, and that the deformed code remains LDPC. 
These properties are obtained via careful construction of $A_X$, $A_Z$, $\gamma_1$, and $\gamma_2$.

Suppose we have a CSS code $\mathcal{C}$, described by the chain complex $C_\bullet$. Let $\overline{X}_L$ be an $X$-type logical operator with support on the physical qubits $L \subseteq C_1$. 
Suppose we can construct a $2$-dimensional chain complex, $D_\bullet$, and a chain map $\gamma_\bullet$ such that (i) the following diagram commutes
\begin{equation}
    \begin{tikzcd}
        & \arrow[d, "\gamma_1"] D_1 \arrow[r, "A_X^T"] & \arrow[d, "\gamma_0"] D_0 \arrow[r, "A_Z"] & D_{-1} \\
        C_2 \arrow[r, "H_X^T"] & C_1 \arrow[r, "H_Z"] & C_0 &
    \end{tikzcd},
\end{equation}
and (ii) we have that $\im \gamma_1 = L$. 
Then the code obtained by taking the total complex includes $L$ in the stabilizer group and hence performs a measurement of $L$. 
Furthermore, if the chain complex $D_\bullet$ is $\varepsilon$-systolic expanding with $\varepsilon \geq 1$, this procedure is distance preserving~\cite{Ide2025,Williamson2024}.

\subsection{Surgery distance}

Suppose we have a CSS code and we want to measure an element $\vec{s} \in  C_2$ such that $ H_X^T \vec{s} = O$, for some operator $O \in C_1$. 
From a physical perspective, the result of measuring the syndrome $C_2$ is a vector $\vec{v}_{\text{syndrome}} \in C_2$, and the measured value of $\vec{s}$ is thus $s_\text{observed} =  \vec{s} \cdot \vec{v}_{\text{syndrome}}  \in \F_2$, where $\vec{v}\cdot \vec{w}$ denotes the dot product of binary vectors. 
In case $\vec{v}_{\text{syndrome}}$ is affected by a measurement error $\vec{e} \in C_2$, the actual measured value is $s_\text{observed}' = s_\text{observed} + \vec{s} \cdot \vec{e} $. The observed value is thus affected if and only if $\vec{s} \cdot \vec{e}\neq 0$ \footnote{Since we work on $\F_2$ we will express this anti-commutation relation as simply $\vec{s} \cdot \vec{e}=1$.}. This error is undetectable if and only if $\vec{v}_{\text{syndrome}} +\vec{e} \in \im H_X$, or, equivalently $\vec{e} \in \im H_X$. This observation motivates the following definition.
\begin{definition}[Surgery distance]
\label{def:surgery-distance}
    The surgery distance for the measurement of an element $\vec{s} \in C_2$ is defined as
    \[
    d_{s} = \min_{\vec{e} \in \im H_X, \vec{s} \cdot \vec{e} = 1} \norm{\vec{e}}
    \]
\end{definition}
\begin{proposition}
    \label{prop:surgery-distance}
     The smallest undetectable syndrome error obeys $\norm{\vec{e}} \geq d_s$. 
\end{proposition}
\begin{proof}
    If $\norm{\vec{e}} < d_s$, then either it does not affect $\vec{s}$, or it can easily be detected since $\vec{e} \not \in \im H_X$.
\end{proof}
The surgery distance is often approximated or estimated by using the meta-check distance~\cite{Hillmann2024}. 
Suppose the code has a meta-check matrix $M_X$, then $d_s$ is lower bounded by $d_M$, the norm of the smallest element not in $\ker M_X$. However this approach often has the limitation that for LDPC codes $d_M \in \mathcal{O}(1)$. 

\subsection{Related work}
There has been an extensive line of work on schemes for fault-tolerant logical Pauli measurements with QLDPC codes since its initiation in Ref.~\cite{Cohen2022}.
We present a brief summary of works that apply similar techniques as used in this manuscript and
for a more complete and detailed overview refer to Section 3.2 of Ref.~\cite{he2025}.

A scheme that improves upon the qubit overhead of the initial protocol from Ref.~\cite{Cohen2022} has been proposed recently for CSS codes in Ref.~\cite{Ide2025}.
There, the ancilla complex is equivalent to Def.~\ref{def:surgery-patch} and is constructed using the cone complex of the complexes corresponding to the initial code and an auxiliary classical code defined by a graph.
A similar protocol that is slightly more generally applicable has been presented in Ref.~\cite{Williamson2024}.
In both works the authors show that to measure a certain logical Pauli operator, an auxiliary graph that fulfils certain properties can be constructed and used to deform the initial code.
The value of the logical operator can then be inferred by a standard gauge fixing procedure on the deformed code.
To maintain the distance of the deformed code, both constructions rely on edge expansion of the auxiliary graph, as opposed to constructing a hypergraph product code with high enough distance as in Ref.~\cite{Cohen2022}.

Moreover, in Ref.~\cite{he2025} the authors present a construction that builds a single ancilla system that can be used to measure any chosen logical operator of the code by varying how the ancilla system is attached to the code block.
Careful construction of the ancilla system and ensuring it has high enough expansion properties allow them to prove the distance and (phenomenological) fault-tolerance of the protocol.

Similarly to this work, to measure arbitrary single and multi-qubit operators of a CSS code, the homomorphic measurement framework~\cite{Huang2023} constructs an ancilla QLDPC code together with a morphism to the initial code block to perform single-shot logical measurements in the vein of Steane error correction.

Finally, methods building on~\cite{Ide2025,Williamson2024,zhang2025time} to separate overlapping logical operators have been developed, enabling parallel measurements for general QLDPC codes~\cite{cowtan2025parallel}.

\section{Fast surgery} \label{sec:surgery}

As previously discussed in Section \ref{subsec:generalized-surgery}, the procedure of surgery on a code can naturally be formalised through homological algebra. In what follows, we will assume we are given $C_\bullet$, a code complex that we wish to perform surgery on; along with a second code complex $D_\bullet$, and a chain homomorphism $\gamma_\bullet$ obeying $\gamma_i : D_i \rightarrow C_i$. 
We formally lay out the surgery procedure in Algorithm \ref{alg:single-shot-surgery}. Subsequently, Subsection \ref{subsec:correctness} is dedicated to proving the correctness of Algorithm \ref{alg:single-shot-surgery}; while Subsections \ref{subsec:distance} and \ref{subsec:fault-tolerance} deal with fault-tolerance guarantees.

In the name of definiteness, we focus on the case of measuring $X$-type logical operators -- the results can easily be ported to the measurement of $Z$-type logical operators. 
The main result of this section can now be stated as follows.

\begin{theorem}
    \label{thm:main-theorem}
    Let $C_\bullet, D_\bullet$ be chain complexes, and $\gamma_\bullet$ a chain homomorphism from $D_\bullet$ to $C_\bullet$. Assume these objects satisfy the following:
    \begin{enumerate}
        \item
        ${\partial_1^D}, {\partial_2^C}, {\partial_1^C}, \gamma_1$ and their transpose are $\omega$-bounded, for some $\omega$
        \item  $D_1$ is $\epsilon$-systolic expanding, for some $\epsilon \geq \omega$
    \end{enumerate}
    Then the surgery procedure given by Algorithm \ref{alg:single-shot-surgery} measures the logical subspace $\gamma_1\left( H_1\left( D_\bullet \right) \right)$, and has fault distance at least $d_Z(D)/\omega$.
\end{theorem}
\begin{proof}
    Corollary of Theorems \ref{thm:surgery-algorithm} and \ref{thm:fault-distance-algorithm}.
\end{proof}

Let $\{e_1, \dots, e_{|D_0|}\}$ denote the standard computational basis for $D_0$.
These vectors can be mapped to Pauli operators in $\cP^{\otimes |D_0|}$ under the standard identification $\mu_X(e_i) = X_i$.
Then, for any vector $u = \sum_i u_i \cdot e_i \in D_0$, this mapping extends linearly, i.e. $\mu_X(u) = \prod_{i: u_i = 1} X_i$. Similarly, any vector $c \in C_1$ can be sent to $\cP^{\otimes |C_1|}$, and we define $\mu_X(u \oplus c) = \mu_X(u) \otimes \mu_X(c) \in \cP^{\otimes |D_0|} \otimes \cP^{\otimes |C_1|}$. 
Finally, we will also denote $\{v_1, \dots, v_{|D_1|}\}$ the standard computational basis for $D_1$.

With this notation, we can arbitrarily pick a pre-image of $\gamma_1 (H_1(D_{\bullet}))$; this step is essential due to the surgery measuring \emph{multiple} operators at the same time.
Write $\{h_l\}_{l=1}^{k_E}$ a selected basis of $\gamma_1 (H_1(D_{\bullet}))$, then for every $h_l$ it is possible to find $\fv_l \subset \{v_1, \dots, v_{|D_1|}\}$ such that $\sum_v \gamma_1(v) = h_l$.
The surgery protocol is then described by Algorithm \ref{alg:single-shot-surgery}.
\begin{algorithm}[H]
    \caption{Fast surgery}
    \label{alg:single-shot-surgery}
    \begin{algorithmic}
        \Require Chain complexes $C_\bullet, D_\bullet$, a chain homomorphism $\gamma_\bullet: D_\bullet \rightarrow C_\bullet$, an input code state $\ket{\Psi}$ of $C_\bullet$, and a selected basis $\{h_l\}_{l=1}^{k_E}$ of $\gamma_1 (H_1(D_{\bullet}))$
        \Ensure The result $\{\sigma_l\}_l$ of measuring $\{h_l\}_{l=1}^{k_E}$; and the post-measurement code state $\prod_l(\idty{} + \sigma_l \cdot \mu_X(h_l))\ket{\Psi}$
        \State $\{\sigma_{l}\}_{l=1}^{k_E} \gets \{1\}$
        \State $\{\omega_{e_i}\}_{i=1}^{i=|D_0|} \gets \{1\}$
        \State $\ket{\Psi} \gets \ket{\Psi}\otimes \ket{0}^{\otimes |D_0|}$ 
        \Comment{Initialize the ancillary register}
        \ForEach{$ v \in \{v_1, \dots, v_{|D_1|}\}$}
            \State $A_v \gets \mu_X((\partial^D_1 + \gamma_1 )v)$
        	\State Measure $A_v$ on $\ket{\Psi}$
            \State $\varepsilon_v \gets $ Measurement result 
            \Comment{Measurement result is $\pm 1$}
            \State $\ket{\Psi} \gets \frac{1}{2}(\idty{}+\varepsilon_v A_v)\ket{\Psi}$
            \Comment{Post-measurement state}
            \ForEach{$ l \in [1, \dots , k_E]$}
                \If{$v \in \fv_l$}
                    \State $\sigma_l \gets \varepsilon_v \cdot \sigma_l$
                \EndIf
            \EndFor
        \EndFor
        \ForEach{$ e \in \{e_1, \dots, e_{|D_0|}\}$}
    	\State Measure $Z_e$ on $\ket{\Psi}$
            \State $\omega _e \gets $ Measurement result
            \Comment{Measurement result is $\pm 1$}
            \State $\ket{\Psi} \gets \frac{1}{2}(\idty{}+\omega_e Z_e)\ket{\Psi}$
            \Comment{Post-measurement state}
        \EndFor
        \State Pick $v' \in D_1$ such that $\partial^D_1v' = (\omega_{e_i})_i$ \Comment{By construction $(\omega_{e_i})_i \in \im \partial^D_1$ }
        \State $\ket{\Psi} \gets \mu_X((\partial^D_1 + \gamma_1 )v')\ket{\Psi}$
        \State Discard ancillary register
    \end{algorithmic}
\end{algorithm}

\subsection{Proof of correctness}
\label{subsec:correctness}

For the procedure given in Algorithm \ref{alg:single-shot-surgery} to be correct, we will want to ensure that we are exactly measuring the subspace $\gamma_1\left( H_1\left( D_\bullet \right) \right)$, which we demonstrate here.
Remember that lattice surgery is, fundamentally, a code deformation process $C_\bullet \rightsquigarrow C_\bullet' \rightsquigarrow C_\bullet$, where the code $C_\bullet'$ happens to contain the logicals of interest as stabilizers.
As such, understanding the properties of $C_\bullet'$ will take a critical place in our argument.

In the present case, $C_\bullet'$ will be uniquely defined by the combination of $C_\bullet, D_\bullet$ and $\gamma_\bullet$.
We begin by a remark that will save us some cumbersome bookkeeping: there always exists $(D_{-1}, \partial^D_{0})$ such that $\ker(\partial^D_{0}) / \im(H'_Z) = \idty$.
Specifically we set $D_{-1} = D_0 / \im\left(H'_Z\right)$ and $\partial_0^D$ to the quotient map from $D_0$ to $D_{-1}$.
With this new module, we obtain the following commuting diagram:
\begin{equation} \label{eq:hom_map2}
E_\bullet = 
\begin{tikzcd}[cells={nodes={minimum height=2em}}]
    D_2 \arrow[r,"{H'}^T_X"] \arrow[d, "\gamma_2"] &  D_1 \arrow[r,"{H'}_Z"] \arrow[d,"\gamma_1"] &  D_0 \arrow[d,"\gamma_0"] \arrow[r,"\partial^D_{0}"] & D_{-1} \\
    C_2 \arrow[r,"H_X^T"] & C_1 \arrow[r, "H_Z"] & C_0
\end{tikzcd}
\end{equation}
Throughout this section we will find it practical to also label this complex as the following double complex (indeed, observe that $\partial^h \partial^v = \partial^v \partial^h$):
\begin{equation} \label{eq:double-complex}
\begin{tikzcd}[cells={nodes={minimum height=2em}}]
    E_{2,1} \arrow[r,"\partial^h_{2,1}"] \arrow[d, "\partial^v_{2,1}"] &  E_{1,1} \arrow[r,"\partial^h_{1,1}"] \arrow[d,"\partial^v_{1,1}"] &  E_{0,1} \arrow[d,"\partial^v_{0,1}"] \arrow[r,"\partial^h_{0,1}"] & E_{-1,1} \\
    E_{2,0} \arrow[r,"\partial^h_{2,0}"] & E_{1,0} \arrow[r, "\partial^h_{1,0}"] & E_{0,0}
\end{tikzcd}
\end{equation}
This dual complex yields a new chain complex given by the total complex $C_\bullet' = \tot (E)_\bullet$ \cite{breuckmann2021balanced}, which is a new, \emph{deformed} quantum code. 
Elements of the code $D_\bullet$ are shifted, so that qubits in $D_\bullet$ are associated with $X$-type checks in $\tot (E)_\bullet$, $Z$-type checks in $D_\bullet$ are associated with qubits in $\tot (E)_\bullet$, and $Z$-type meta-checks in $D_\bullet$ are associated with $Z$-type checks in $\tot (E)_\bullet$. 
In particular, $X$-type checks in $D_\bullet$ will be associated with $X$-type meta-checks in $\tot (E)_\bullet$, which will make the logical measurement being performed robust to syndrome errors.
\begin{equation}
\tot (E)_\bullet = 
\begin{tikzcd}[ampersand replacement=\&]
    D_2 \arrow{r}{\begin{pmatrix} {H'}^T_X \\ \gamma_2 \end{pmatrix}}
    \& D_1 \oplus C_2 \arrow{r}{\begin{pmatrix} H'_Z & 0 \\ \gamma_1 & H_X^T \end{pmatrix}}
    \& D_0 \oplus C_1 \arrow{r}{\begin{pmatrix} \gamma_0 & H_Z \end{pmatrix}}
    \& C_0
\end{tikzcd}
\end{equation}
\paragraph*{Remark:} The construction can be straightforwardly extended to the case when the code $C_\bullet$ has $X$-type and $Z$-type meta-checks, in which case it is described by a $4$-dimensional chain complex.
The only additional work is to construct a suitable $\gamma_{-1}: D_{-1} \to C_{-1}$. In fact $\gamma_{-1}$ is uniquely given by the map $\partial_0^C \circ \gamma_0 \circ \left(\partial_0^D\right)^{-1}$, where $\left(\partial_0^D\right)^{-1}$ is the right inverse of $\partial_0^D$.
Such an inverse exists because we defined $\partial_0^D$ such that it has full row rank. Then $\tot (E)_\bullet$ is constructed from the total complex of the resulting dual complex, as above.

In the code $\tot (E)_\bullet$, the logical operators from $C_\bullet$ that are in the image of $\gamma_1$ become stabilizers -- and can thus be readily measured. We will make this statement formal by obtaining an explicit basis for the set of logical $H_1(\tot E)$. As a first step we compute the exact number of logicals in $H_1(\tot E)$ \footnote{It can be noted that if $\gamma_1$ is full rank on $H_1(D_\bullet)$, then Lemma \ref{lemma:dimension} gives the naturally expected $k_E = k_C - k_D $.}.

\begin{lemma} \label{lemma:dimension}
    \[
    k_E \equiv \dim H_1(\tot E) = k_C - \dim \gamma_1 (H_1(D_{\bullet}))
    \]
\end{lemma}

\begin{proof}
    From \cite{breuckmann2021balanced} (Theorem 2), we have:
    \[H_n(\tot(E)) \cong \bigoplus_{p+q =n} H_q (H_p(E_{\bullet, \bullet}, \partial^h), \partial^v)\]

    Let's unravel this expression: $ H_{\bullet} (H_{\bullet}(E_{\bullet, \bullet}, \partial^h), \partial^v)$ denotes a complex onto itself.
    We will go through the process of iteratively building it up so that the expression for $H_n(\tot(E))$ makes sense.
    First, $H_p(E_{\bullet, q}, \partial^h)$ denote the complex obtained from taking the homology groups along the $\partial^h$ axis.
    For the sake of clarity, we explicitly have:

    \[
     H_{\bullet}(E_{\bullet, \bullet}, \partial^h) = \begin{tikzcd}[ampersand replacement=\&]
        H_{2}(E_{\bullet, 1}) \arrow[r,"\partial^h_{2,1}"] \arrow[d,"\partial^v_{2,1}"] 
        \&  H_{1}(E_{\bullet, 1})\arrow[r,"\partial^h_{1,1}"] \arrow[d,"\partial^v_{1,1}"] 
        \& H_{0}(E_{\bullet, 1}) \arrow[r,"\partial^h_{0,1}"] \arrow[d,"\partial^v_{0,1}"] 
        \& H_{-1}(E_{\bullet, 1}) 
        \\
        H_{2}(E_{\bullet, 0}) \arrow[r,"\partial^h_{2,0}"] \&  H_{1}(E_{\bullet, 0}) \arrow[r,"\partial^h_{1,0}"] 
        \& H_{0}(E_{\bullet, 0})
    \end{tikzcd}
    \]
    Where: 
    \begin{align*}
        & H_{2}(E_{\bullet, 1}) = \ker(\partial^h_{2,1}) & H_{1}(E_{\bullet, 1}) = \ker(\partial^h_{1,1}) / \im(\partial^h_{2,1}) = k_D \\ 
        &  H_{0}(E_{\bullet, 1}) = \ker(\partial^h_{0,1})/ \im(\partial^h_{1,1}) =0 & H_{-1}(E_{\bullet, 1}) = E_{-1,1} / \im(\partial^h_{0,1}) \\
        & H_{2}(E_{\bullet, 0}) = \ker(\partial^h_{2,0}) & H_{1}(E_{\bullet, 0}) = \ker(\partial^h_{0,0})/ \im(\partial^h_{1,0}) = k_C  \\
        & H_{0}(E_{\bullet, 0}) = E_{0,0}/\im(\partial^h_{1,0})
    \end{align*}

    Note that the horizontal arrows are trivialised, they now send every element to $0$.
    To keep the notation light, we will relabel this complex:

    \[
     E^1_{\bullet, \bullet} = \begin{tikzcd}[ampersand replacement=\&]
        E^1_{2,1} \arrow[r,"\partial^h_{2,1}"] \arrow[d,"\partial^v_{2,1}"] 
        \& E^1_{1,1}\arrow[r,"\partial^h_{1,1}"] \arrow[d,"\partial^v_{1,1}"] 
        \& E^1_{0,1} \arrow[r,"\partial^h_{0,1}"] \arrow[d,"\partial^v_{0,1}"]
        \& E^1_{-1,1}
        \\
        E^1_{2,0}  \arrow[r,"\partial^h_{2,0}"] 
        \& E^1_{1,0}  \arrow[r,"\partial^h_{1,0}"] 
        \& E^1_{0,0}
    \end{tikzcd}
    \]

    The final step is to compute $H_{\bullet} (H_{\bullet}(E_{\bullet, \bullet}, \partial^h), \partial^v) =  H_{\bullet} (E^1_{\bullet, \bullet}, \partial^v)$

    \[
    H_{\bullet} (E^1_{\bullet, \bullet}, \partial^v) = \begin{tikzcd}[ampersand replacement=\&]
        H_{1}(E^1_{2, \bullet}) \arrow[r,"\partial^h_{2,1}"] \arrow[d,"\partial^v_{2,1}"] 
        \& H_{1}(E^1_{1,\bullet})\arrow[r,"\partial^h_{1,1}"] \arrow[d,"\partial^v_{1,1}"] 
        \& H_{1}(E^1_{0,\bullet}) \arrow[r,"\partial^h_{0,1}"] \arrow[d,"\partial^v_{0,1}"] 
        \& H_{1}(E^1_{-1,\bullet}) 
        \\
        H_{0}(E^1_{2, \bullet}) \arrow[r,"\partial^h_{2,0}"] 
        \& H_{0}(E^1_{1,\bullet}) \arrow[r,"\partial^h_{1,0}"] 
        \& H_{0}(E^1_{0,\bullet})
    \end{tikzcd}
    \]
    According to the previous formula, we have:
    \begin{equation}
        H_1(\tot E) \cong H_{0}(E^1_{1,\bullet}) \oplus H_{1}(E^1_{0,\bullet})
    \end{equation}
    As this expression involves only two modules, $H_{0}(E^1_{1,\bullet})$ and $H_{1}(E^1_{0,\bullet})$, we will focus uniquely on them. We first have:
    \[
    H_{0}(E^1_{1,\bullet}) = E^1_{1,0} /  \partial^v_{1,1} (E^1_{1,1})
    \]
    Since $\gamma_1$ is a chain homomorphism, we have $\partial^v_{1,1} (E^1_{1,1}) \subset E^1_{1,0}$, which gives:
    \begin{align}
        \dim H_{0}(E^1_{1,\bullet}) &= \dim E^1_{1,0} - \dim  \partial^v_{1,1} (E^1_{1,1}) \\ 
        &= k_C - \dim  \partial^v_{1,1} (E^1_{1,1}) \\
        &= k_C - \dim \gamma_1 (H_1(D_{\bullet}))
    \end{align}
    Now, moving on to $H_{1}(E^1_{0,\bullet})$, note that $\dim E^1_{0,1} = \dim H_{0}(E_{\bullet, 1}) = 0$. Its homology group can then only be trivial: $H_{1}(E^1_{0,\bullet}) = \idty$. We now have all the elements to conclude:
    \[
    \dim H_1(\tot E) =  \dim H_{0}(E^1_{1,\bullet}) =  k_C - \dim \gamma_1 (H_1(D_{\bullet}))
    \]
    
\end{proof}

The number of logicals $k_E$ can now be leveraged to justify that a basis for $H_1(C_\bullet ) / \gamma_1 (H_1(D_{\bullet}))$ is sufficient to find a basis for the entirety of $H_1(\tot E)$.

\begin{lemma}
    \label{lem:homology-basis}
    Let $\{\vec{u}\}_u$ be a basis for $H_1(C_\bullet ) / \gamma_1 (H_1(D_{\bullet}))$, then it is also a basis for $H_1(\tot E)$. 
\end{lemma}
\begin{proof}
    
    We can readily verify that $Z_1(C_\bullet) \subset Z_1(\tot E_\bullet)$. 
    Therefore the group $(Z_1(C_\bullet) / B_1(C_\bullet) ) / \gamma_1 (H_1(D_{\bullet})) = H_1(C_\bullet ) / \gamma_1 (H_1(D_{\bullet})) $ maps to independent equivalence classes in $H_1(\tot E_\bullet)$. Moreover, the dimension of that group also happens to match that of $\dim  H_1(\tot E_\bullet)$:
    \begin{align}
        \dim H_1(C_\bullet ) / \gamma_1 (H_1(D_{\bullet})) &= k_C - \dim \gamma_1 (H_1(D_{\bullet})) \\ &= \dim  H_1(\tot E_\bullet)
    \end{align}

    We conclude that a basis for $ H_1(C_\bullet ) / \gamma_1 (H_1(D_{\bullet}))$ forms a basis for $H_1(\tot E)$.
    
\end{proof}

\begin{theorem}
    \label{thm:surgery-algorithm}
    The code deformation defined by $C_\bullet \to \tot (E)_\bullet \to C_\bullet$ given by Algorithm \ref{alg:single-shot-surgery} measures the logical operators given by $\gamma_1\left( H_1\left( D_\bullet \right) \right)$
\end{theorem}

\begin{proof}

    \begin{enumerate}
        \item $\gamma_1\left( H_1\left( D_\bullet \right) \right) \subset \im \tot (E)_2$, or equivalently, the operators in $\gamma_1\left( H_1\left( D_\bullet \right) \right)$ are in the image of the stabilizers $\tot (E)_2$
        \item From Lemma \ref{lem:homology-basis}, we can read off that no other logical in $H_1(C_\bullet )$ is in $\im \tot (E)_2$, i.e. they are not measured.
    \end{enumerate}
    
\end{proof}

\subsection{Proof of distance}
\label{subsec:distance}

In this section we demonstrate that the code given by the total complex $\tot (E)_\bullet$ is distance preserving, assuming several conditions on $D_\bullet$ and the maps $\gamma_\bullet$ are satisfied.

\begin{theorem}
    \label{thm:distances-of-surgery}
    Let $C_\bullet, D_\bullet$ be chain complexes, and $\gamma_\bullet$ a chain homomorphism from $D_\bullet$ to $C_\bullet$. We assume these objects satisfy the following:
    \begin{enumerate}
        \item  \label{prop:bounded-maps}${\partial_1^D}, {\partial_2^C}, {\partial_1^C}, \gamma_1$ and their transpose are $\omega$-bounded, for some $\omega$
        \item  $D_1$ is $\epsilon$-systolic expanding, for some $\epsilon \geq \omega$
    \end{enumerate}
   Then:
    \begin{enumerate}
        \item $d_X(\tot E) \geq  d_X(C)$
        \item $d_Z(\tot E) \geq d_Z(C)$
        \item For any $s \in H_1(D_\bullet )$, $d_s \geq d_Z(D)$ 
    \end{enumerate}
\end{theorem}
\begin{proof}
    Note Lemma \ref{lemma:dx-distance}, \ref{lemma:dz-distance} and \ref{lemma:lower-bound-surgery-distance}.
\end{proof}

\begin{lemma}
\label{lemma:dx-distance}
    $d_X(\tot E) \geq \min (1, \frac{\epsilon}{\omega})\cdot d_X(C)$.
    In particular, if $D_1$ is $\epsilon$-systolic expanding for some $\epsilon \geq \omega$, then $d_X(\tot E) \geq d_X(C)$.
\end{lemma}
\begin{proof}
    Denote $\cL' \subset C_1\oplus D_0$ the space spanned by the basis given by Lemma \ref{lem:homology-basis}. We can then compute the distance as:
    \[
    d_X( \tot E) = \min_{L \in \cL' \setminus 0} \norm{L + \im (\partial_2^C) + \im (\gamma_1+\partial_1^D)}
    \]
    Ideally, we would like to lower bound $d_X( \tot E)$ in terms of $d_X(C)$. To that effect, we will describe to what extent can $w \in \cL' \setminus 0 + \im (\partial_2^C)$ see its weight be reduced by the stabilizers in $\tot (E)_2$.
    First, from the systolic expansion of $\partial_1^D$, we have the following:
     \begin{align*}
        \forall u \not\in \ker(\partial_1^D), \quad \norm{w + (\gamma_1+\partial_1^D)u} &\geq  \norm{w + \gamma_1u} + \norm{\partial_1^D u}\\
        & \geq \norm{w + \gamma_1u} + \epsilon \cdot \min_{P \in \ker(\partial_1^D)} \norm{u+P}
    \end{align*}
    
    Note that only the case where $u \not\in \ker(\partial_1^D)$ is relevant here, otherwise $\partial^C_1 \circ \gamma_1 (u) = 0$, and $w + \gamma_1 u \in \ker(\partial^C_1)$, where it becomes trivial to lower bound its weight.
    We write $P_u$ an operator that saturates the minimisation for $u$, then:
    \[
    \norm{w + (\gamma_1+\partial_1^D)u} \geq \norm{w + \gamma_1(u+P_u) + \gamma_1(P_u)} + \epsilon  \norm{u+P_u}
    \]
    Remember that $\gamma_1$ is a chain homomorphism, hence $P'_u \equiv \gamma_1 (P_u)  \in \ker(\partial^C_1)$. 
    We write $w' \equiv w + P'_u$, and $u' \equiv u + P_u$; note that from the definition of $\cL'$, $w' \in \cL' \setminus 0 + \im (\partial_2^C)$ too. The above expression now simplifies to the following:
    \[
    \norm{w + (\gamma_1+\partial_1^D)u} \geq \norm{w' + \gamma_1(u') } + \epsilon  \norm{u'}
    \]
    We will now address the case $\epsilon \geq \omega$ first, and then that of $\epsilon < \omega$ separately. Remember that $w' \in \cL' \setminus 0 + \im (\partial_2^C)$, and all the elements of $\cL' \setminus 0 + \im (\partial_2^C)$ satisfy $\norm{w'} \geq d_X(C)$. Hence:
    \begin{align*}
        \norm{w + (\gamma_1+\partial_1^D)u} &\geq \norm{w' + \gamma_1(u') } + \epsilon  \norm{u'} \\ & \geq \norm{w' } - \norm{\gamma_1(u')} + \epsilon  \norm{u'} \\
        & \geq \norm{w' } - \omega \norm{u'} + \epsilon  \norm{u'} \\ &\geq \norm{w' } \geq d_X(C)
    \end{align*}
    We now proceed with the case $\epsilon < \omega$. From the boundedness of $\gamma_1$, we have that $\norm{u'} \geq \norm{\gamma_1 u'}/\omega$, and from the reverse triangle inequality $\norm{\gamma_1 u'} \geq \norm{w'} - \norm{w' + \gamma_1 u'}$. Combining these bounds gets us:
   \begin{align*}
       \norm{w + (\gamma_1+\partial_1^D)u} &\geq \norm{w' + \gamma_1(u') } + \epsilon \frac{\norm{w'} - \norm{w' + \gamma_1 u'}}{\omega}\\
       & = \frac{\epsilon}{\omega}\norm{w'} + (1-\frac{\epsilon}{\omega})\norm{w' + \gamma_1 u'} \\
       & \geq \frac{\epsilon}{\omega}\norm{w'} \geq \frac{\epsilon}{\omega} d_X(C)
   \end{align*}

\end{proof}

\begin{lemma}
    \label{lemma:dz-distance}
    $d_Z( \tot E) \geq d_Z(C)$
\end{lemma}
\begin{proof}

    Denote $c \oplus d \in C^1 \oplus D^0$ a non-trivial element of $H^1(\tot E_{\bullet})$.
    It is always true that $c \neq 0$, as $\ker (\delta_0^D) = \im (\delta_{-1}^D)$ and thus otherwise $d$ would then be a trivial element. With this guarantee, established, we can now assert that 
    \[
    c \in \ker(\delta_1^C) \setminus \im (\delta_0^C)
    \]
    In other words, $c$ is a non-trivial element of $H^1(C_\bullet)$, and inherits its distance. This gives us:
    \[
    \norm{c \oplus d} \geq \norm{c} \geq d_Z(C)
    \]
\end{proof}

\begin{lemma}
\label{lemma:lower-bound-surgery-distance}
    For any $s \in H_1(D_\bullet )$,  $d_s \geq d_Z(D)$.
\end{lemma}
\begin{proof}
    First note that:
    \[
    d_s \geq \min_{v \in \ker{{\partial_2^D}^T} \setminus \im{{\partial_1^D}^T} \setminus \im {\gamma_1^T}} \norm{v}
    \]
    Because the measured syndrome commutes with all the elements of $\im{{\partial_1^D}^T} \cup \im {\gamma_1^T}$, and elements not in $\ker{{\partial_2^D}^T}$ are not in $\im H_X$. Note then:
    \[
    \min_{v \in \ker{{\partial_2^D}^T} \setminus \im{{\partial_1^D}^T} \setminus \im {\gamma_1^T}} \norm{v} \geq \min_{v \in \ker{{\partial_2^D}^T} \setminus \im{{\partial_1^D}^T} } \norm{v} = d_Z(D)
    \]
    Which concludes the proof.
\end{proof}

\subsection{Proof of fault-tolerance}
\label{subsec:fault-tolerance}

The surgery distance of Definition \ref{def:surgery-distance} characterises the resilience of $\tot(E)_\bullet$ to measurement-only errors.
However, in a realistic setting, errors might occur on both measurements \emph{and} qubits.
We accommodate this subtlety by generalising the surgery distance to the more general fault-distance.

\begin{definition}
The measurements of the $A_v$'s in Algorithm \ref{alg:single-shot-surgery} are performed on a qubit register affected by $n_{\text{qubit}}$ errors, and the outcome of $n_{\text{measurements}}$ are flipped.
The procedure is said to have \textbf{fault distance} $d$ if the smallest undetectable error that affects the output $\{\sigma_l\}_l$ of Algorithm \ref{alg:single-shot-surgery} has weight $n_{\text{qubit}} + n_{\text{measurements}} = d-1$.
\end{definition}

Algorithm \ref{alg:single-shot-surgery} gives a fault-tolerant procedure if we make the standard assumption that the syndromes in $C_1$ are sufficiently reliable both before and after performing the surgery procedure.
This leads to Theorem \ref{thm:fault-distance-algorithm} guaranteeing the fault-tolerance of Algorithm \ref{alg:single-shot-surgery}.

\begin{theorem}
    \label{thm:fault-distance-algorithm}
     Consider chain complexes $C_\bullet, D_\bullet$, and a chain homomorphism $\gamma_\bullet: D_\bullet \rightarrow C_\bullet$ such that the assumption \ref{prop:bounded-maps} of Theorem \ref{thm:distances-of-surgery} is satisfied. Then, the fault distance of the fault-tolerant application of Algorithm \ref{alg:single-shot-surgery} is lower bounded by $d_Z(D)/\omega$.
\end{theorem}
\begin{proof}
    Let $v_{C_2} \in C_2, v_{D_1} \in D_1$ be the observed syndrome measurements. They each are corrupted by some error vectors:

    \begin{equation}
        v_{C_2} = v_{C_2}' + e_{C_2} , \quad  v_{D_1}= v_{D_1}' + e_{D_1}
    \end{equation}
    Similarly, the qubits, at the time they were measured, were also affected by some errors, which we write $e_{D_0} \in D_0, e_{C_1} \in C_1$.
    The total error $e_{\text{tot}} \equiv e_{D_1} + e_{D_0} + e_{C_2} + e_{C_1} $ leads to the same measured syndrome as the following equivalent error (we push the qubit errors onto the syndrome data):
    \[
    e_{\text{tot}}' = (e_{D_1} + {\partial_1^D}^T e_{D_0} + \gamma_1^T e_{C_1} ) + (e_{C_2} + {\partial_2^C}^T e_{C_1}) \equiv e_{D_1}' + e_{C_2}'
    \]
    Where $\norm{e_{\text{tot}}'} \leq e_{D_1} + \omega e_{D_0} + e_{C_2} + \omega e_{C_1} + \omega e_{C_1}$ due to the $\omega$-boundedness of the respective maps.
    Because we have perfect measurement of $C_2$ pre- and post- surgery, we know exactly what $e_{C_2}'$ is, and we can set it to be $0$.
    
    Finally, Lemma \ref{lemma:lower-bound-surgery-distance} and Proposition \ref{prop:surgery-distance} guarantee that as long as $\norm{e_{D_1}'} < d_Z(D)$, we have $\norm{e_{D_1}'} < d_s$ and no undetectable error has corrupted the output of Algorithm \ref{alg:single-shot-surgery}. The condition $\norm{e_{D_1}'} < d_Z(D)$ always holds as long as $\norm{e_{D_1}} + \norm{e_{D_0}} + \norm{e_{C_1}} < d_Z(D)/\omega $. We conclude that the fault distance is at least $ d_Z(D)/\omega$.
\end{proof}

\paragraph*{Remark:} In the case where the base code $C_\bullet$ is presented with a sufficiently large meta check distance of its own, the theorem above can be adapted to show that $r$ rounds of syndrome measurement amplifies the fault distance to $r \cdot d_Z(D)/\omega$.

\section{Constructing homomorphic complexes} \label{sec:constructing}

\subsection{Expansion boosting with the repetition code} \label{sec:boost}

In this section we describe how, starting with a $D_\bullet$ complex with low or no expansion, we can obtain a new complex with sufficiently boosted expansion that will guarantee the resulting $d_X(\tot E)$ is large.
The key observation is that the property $\norm{\partial^Du'} \geq \epsilon\norm{u'}$ used in Lemma \ref{lemma:dz-distance} can be relaxed to $\norm{\partial^Du'} \geq \epsilon\norm{\gamma_1(u')}$. 
This can be achieved simply by producting $D_\bullet$ with a repetition code of length $d_X(C)$.
The notation we use assumes that $D_\bullet$ exists in the context of the following surgery scheme -- we ignore $D_{-1}$ as it can always be fixed ad-hoc:

\begin{equation}
\begin{tikzcd}[cells={nodes={minimum height=2em}}]
    D_2 \arrow["\partial_2^D"]{r} \arrow[d, "\gamma_2"] &  D_1 \arrow["\partial_1^D"]{r} \arrow[d,"\gamma_1"] &  D_0 \arrow[d,"\gamma_0"] \\
    C_2 \arrow["\partial_2^C"]{r} & C_1 \arrow["\partial_1^C"]{r} & C_0
\end{tikzcd}
\end{equation}

We denote by $R_1 \xrightarrow{\partial_1^R} R_0$ a repetition code of length $l$ -- and we will later fix $l =d_X(C)$.
The boundary map $\partial_1^R$ is picked such that $H_0(R_\bullet) = 0$, i.e. $R_0 = \im \partial_1^R$ \footnote{In this case $\partial_1^R$ is simply the `usual' repetition code; with checks $Z_1Z_2, Z_2Z_3, \dots, Z_{l-1}Z_l$.}.
We adopt the convention where $R_1$ refers to the bits, and $R_0$ to the checks.
As previously hinted at, the new ancilla complex is now $D^{\otimes l}_\bullet \equiv (D\otimes R)_\bullet$, which corresponds to the following complex:

\begin{equation}
    (D\otimes R)_\bullet = \begin{tikzcd}
                                              &                                          &  D_0 \otimes R_1 \arrow["\idty \otimes \partial_1^R"]{rd} &                 \\
                                              &   D_1 \otimes R_1 \arrow["\partial_1^D \otimes \idty"]{ru} \arrow["\idty \otimes \partial_1^R"]{rd}  &                             & D_0 \otimes R_0\\     
        D_2 \otimes R_1 \arrow["\partial_2^D \otimes \idty"]{ru} \arrow["\idty \otimes \partial_1^R"]{rd} &                                          &  D_1 \otimes R_0 \arrow["\partial_1^D \otimes \idty"]{ru} &                 \\
                                              &  D_2 \otimes R_0  \arrow["\partial_2^D \otimes \idty"]{ru}             &\\
    \end{tikzcd}
\end{equation}

This product complex has boundary operator $\partial^{D\otimes R} = \partial^D\otimes \idty + \idty \otimes \partial^R$, and its modules are denoted as follows:

\begin{align*}
    (D\otimes R)_2 &= D_2 \otimes R_1, \quad \\
    (D\otimes R)_1 &= D_1 \otimes R_1 \bigoplus D_2 \otimes R_0, \quad \\
    (D\otimes R)_0 &= D_0 \otimes R_1 \bigoplus D_1 \otimes R_0, \quad \\
    (D\otimes R)_{-1} &= D_0 \otimes R_0, \quad
\end{align*}

With this new ancillary complex, we can update the surgery scheme as follows:

\begin{equation}
\label{eq:repetition-new-ancilla}
\begin{tikzcd}[cells={nodes={minimum height=2em}}]
    (D\otimes R)_2 \arrow["\partial_2^{D\otimes R}"]{r} \arrow[d, "\gamma'_2"] &  (D\otimes R)_1 \arrow["\partial_1^{D\otimes R}"]{r} \arrow[d,"\gamma'_1"] &  (D\otimes R)_0 \arrow[d,"\gamma'_0"] \\
    C_2 \arrow["\partial_2^C"]{r} & C_1 \arrow["\partial_1^C"]{r} & C_0
\end{tikzcd}
\end{equation}

To define the new homomorphisms $\gamma'_\bullet: (D\otimes R)_\bullet\to C_\bullet$ we fix a distinguished bit $r_\star \in R_1$ at one of the two endpoints of the repetition code and define a projection $\pi: R_1 \rightarrow \F_2$ by $\pi(r_\star) = 1$ and $\pi(r) = 0$ for any basis bit $r \neq r_\star$.
Meanwhile, $\gamma'_\bullet$ sends $D_\bullet \otimes R_0$ to $0$. With these conventions we can make the definition of the chain precise:
\begin{align*}
  &\gamma'_2(d\otimes r)=\pi(r)\,\gamma_2(d)\quad\text{for }d\in D_2,\ r\in R_1, \\
  &\gamma'_1(d\otimes r)=\pi(r)\,\gamma_1(d)\quad\text{for }d\in D_1,\ r\in R_1,\quad\text{and }\gamma'_2(d\otimes r)=0\text{ for } d \in D_2, r\in R_0,\\
  &\gamma'_0(d\otimes r)=\pi(r)\,\gamma_0(d)\quad\text{for }d\in D_0,\ r\in R_1,\quad\text{and }\gamma'_1(d\otimes r)=0\text{ for } d \in D_1, r\in R_0.\\
\end{align*}

\begin{lemma}
    \label{lemma:product-expansion}
    Let $u \in (D\otimes R)_1$, and let $u'=u+z$ be a representative of the coset $u+Z_1(D\otimes R)$ that minimises $\norm{\gamma_1(u')}$. Then
    \[
    \norm{\partial u'} \geq \min (\norm{\gamma'_1(u')}/\omega, l)
    \]
\end{lemma}

\begin{proof}
  \textit{Step 1:} We start by eliminating the $D_2\otimes R_0$ component of $u'$ without changing $\partial u'$ or $\gamma'_1(u')$.
  Write $u'$ as $u'=u'_{11}\oplus u'_{20}$ with $u'_{11}\in D_1\otimes R_1$ and $u'_{20}\in D_2\otimes R_0$. Because $R_0=\operatorname{im}\partial_1^R$, for each element $d_2\otimes r_0$ in $u'_{20}$ there exists $r_1\in R_1$ with $\partial_1^R r_1=r_0$. 
  Further, since the all-ones vector belongs to $\ker\partial_1^R$, it is possible to pick $r_1$ such that $\pi(r_1)=0$ -- otherwise we can `clean' the $r_\star$ component by considering $r_1 \rightarrow r_1 + \mathbf{1}$. We then set
  \[
    u'' \equiv u' + \partial(d_2\otimes r_1) = u' + (\partial_1^D d_2\otimes r_1) + (d_2\otimes \partial_1^R r_1)
  \]
  This cancels the $d_2\otimes r_0$ term in $u'_{20}$, so after applying this to every element of $u'_{20}$ we obtain a representative (still denoted $u''$) with no $D_2\otimes R_0$ component, i.e. $u''\in D_1\otimes R_1$.
  Moreover, as $\partial(d_2\otimes r_1)$ is a stabilizer then $\partial u''=\partial u'$; and by construction $\pi(r_1)=0$ gives $\gamma'_1\big(\partial_1^D d_2\otimes r_1\big)=0$, hence $\gamma'_1(u'')=\gamma'_1(u')$.

  \textit{Step 2:} We now proceed to lower bound $\norm{\partial u''}$ for $u''\in D_1\otimes R_1$.
  Write $u''=\sum_{j=1}^l d_j\otimes e_j$, where $\{e_j\}_j$ is the standard basis of $R_1$ (bits along the length-$l$ repetition code).
  We denote by $j_\star$ the index of the distinguished bit, i.e. $\pi(e_{j_\star}) = 1$, and $\gamma'_1(u'') = \gamma_1(d_{j_\star})$ -- note that we can always assume $\norm{d_{j_\star}} > 0$, otherwise the lemma is immediately proved. Then
  \[
    \partial u'' 
    = (\partial_1^D\otimes\mathrm{id})(u'') + (\mathrm{id}\otimes\partial_1^R)(u'')
    = \sum_{j} \partial_1^D d_j\otimes e_j \oplus \sum_{j} d_j\otimes \partial_1^R e_j
  \]
  Write $\Delta = \max_j \norm{d_{j_\star} - d_j}$. At this point the rest of the argument can be broken down into two cases:
  \begin{enumerate}
      \item $\Delta < \norm{d_{j_\star}}$. In which case, $\norm{\partial_1^D d_j} \geq 1$ for every $e_j$. This yields $\norm{\partial u''} \geq \norm{\sum_j \partial_1^D d_j\otimes e_j} \geq 1 * l$
      \item $\Delta \geq \norm{d_{j_\star}}$. In which case, $\norm{\partial u''} \geq \norm{\sum_{j} d_j\otimes \partial_1^R e_j} = \sum_{j=1}^{l-1} \norm{d_j - d_{j+1}} \geq \Delta \geq \norm{d_{j_\star}}$.
  \end{enumerate}
  By the $\omega$-boundedness of $\gamma_1$, we have:
  \[
    \norm{\partial u'} = \norm{\partial u''} \geq  \min ( \norm{\gamma'_1(u'')}/\omega,\ l )
    = \min ( \norm{\gamma'_1(u')}/\omega,\ l )
  \]
\end{proof}

\begin{theorem}
    \label{thm:distances-of-surgery-repetition}
    Let $C_\bullet, D_\bullet$ be chain complexes, and $\gamma_\bullet$ a chain morphism from $D_\bullet$ to $C_\bullet$. We assume these objects satisfy the following:
    \begin{enumerate}
        \item  ${\partial_1^D}, {\partial_2^C}, {\partial_1^C}, \gamma_1$ and their transpose are $\omega$-bounded, for some $\omega$
    \end{enumerate}
   Then the new homomorphisms described in Equation \ref{eq:repetition-new-ancilla} with $l=d_X(C)$ yield $\tot E$ satisfying the following properties:
    \begin{enumerate}
        \item \label{property:repetition-0} $\partial_1^{D\otimes R}$ and its transpose is $2\omega$-bounded; $\gamma_1'$ and its transpose are $\omega$-bounded.
        \item \label{property:repetition-1} $d_X(\tot E) \geq \frac{1}{\omega} d_X(C)$
        \item \label{property:repetition-2} $d_Z(\tot E) \geq d_Z(C)$
        \item \label{property:repetition-3} For any $s \in H_1(D_\bullet )$, $d_s \geq d_Z(D)$ 
    \end{enumerate}
\end{theorem}
\begin{proof}
    The boundedness of $\partial_1^{D\otimes R},\gamma_1'$ and their transpose is immediate from their definition.
    Properties \ref{property:repetition-2} and \ref{property:repetition-3} follow from Lemmas \ref{lemma:dz-distance} and \ref{lemma:lower-bound-surgery-distance} respectively, and the properties of the hypergraph product \cite{tillich2014hypergraph,zeng2019higher}. 
    
    Property \ref{property:repetition-1} follows from Lemma \ref{lemma:product-expansion}, in a way much similar to that of the proof of Lemma \ref{lemma:dx-distance}.
    Denote $\cL' \subset C_1\oplus (D\otimes R)_0$ the space spanned by the basis given by Lemma \ref{lem:homology-basis}. We can then compute the distance as:
    \[
    d_X( \tot E) = \min_{L \in \cL' \setminus 0} \norm{L + \im (\partial_2^C) + \im (\gamma'_1+\partial_1^{D\otimes R})}
    \]
    We will describe to what extent can $w \in \cL' \setminus 0 + \im (\partial_2^C)$ see its weight be reduced by the stabilizers in $\tot (E)_2$. First, from Lemma \ref{lemma:product-expansion}
     \begin{align*}
        \forall u \not\in \ker(\partial_1^{D\otimes R}), \quad \norm{w + (\gamma'_1+\partial_1^{D\otimes R})u} &\geq  \norm{w + \gamma'_1u} + \norm{\partial_1^{D\otimes R} u}
    \end{align*}
    In the case $\norm{\gamma'_1(u')}/\omega \geq l$, we have $\norm{w + (\gamma'_1+\partial_1^{D\otimes R})u} \geq \norm{w + \gamma'_1u} + l \geq l = d_X(C)$, and the conclusion follows immediately. We thus focus on the opposite situation. Write 
    \begin{align*}
        \norm{w + (\gamma_1+\partial_1^{D\otimes R})u} &\geq  \norm{w + \gamma'_1u} + \frac{1}{\omega}\norm{\gamma'_1(u')} \\
        & = \norm{w + \gamma'_1(u + u' - u')} + \frac{1}{\omega}\norm{\gamma'_1(u')} \\
        &= \norm{(w + \gamma'_1(u - u')) + \gamma_1(u')} + \frac{1}{\omega}\norm{\gamma'_1(u')}  \\
        &\equiv \norm{ w' + \gamma'_1(u')} + \frac{1}{\omega}\norm{\gamma'_1(u')} \\
    \end{align*}
    From the reverse triangle inequality, $\norm{\gamma'_1(u')} \geq \norm{w'} - \norm{w' + \gamma'_1 u'}$, the expression then becomes:
    \begin{align*}
       \norm{w + (\gamma_1+\partial_1^{D\otimes R})u} &\geq \norm{w' + \gamma'_1(u') } + \frac{\norm{w'} - \norm{w' + \gamma'_1 u'}}{\omega}\\
       & = \frac{1}{\omega}\norm{w'} + (1-\frac{1}{\omega})\norm{w' + \gamma_1 u'} \\
       & \geq \frac{1}{\omega}\norm{w'} \geq \frac{1}{\omega} d_X(C)
   \end{align*}
\end{proof}

\paragraph{Remark:} This construction yields a time-optimal extremum that contrasts well with the space-optimal extremum presented in \cite{Williamson2024}.

\subsection{Constructing $D_\bullet$}

We now describe a method for creating an ancilla complex $D_{\bullet}$, encoding a single logical qubit that is homomorphic to $C_{\bullet}$ for an arbitrary CSS QLDPC code $C_{\bullet}$.

A straightforward way to obtain a suitable complex ancilla complex $D_\bullet$ for any given code $C_\bullet$ is to construct an ancilla system by using ideas from the parallel logical measurement scheme of~\cite{cowtan2025parallel}, which combines brute-force branching and the measurement scheme from Refs.~\cite{Williamson2024,Ide2025}.
\begin{lemma}[Lemma 2.5 in \cite{cowtan2025parallel}]\label{lem:bfb-no-new-logicals}
    Brute-force branching creates new representatives on each leaf for each logical support $v_i \in \mathcal{I}$ and introduces no new logical operators.
\end{lemma}
\begin{lemma}[Lemma 2.6 in \cite{cowtan2025parallel}, Thm 2 in \cite{zhang2025time}]\label{lem:bfb-distance-pres}
    Brute-force branching preserves the code distance.
\end{lemma}
\begin{lemma}[Lemma 2 \cite{Williamson2024}]
    The distance $d'$ of the deformed code obtained by the gauging logical measurement scheme satisfies $d' \geq d$ if the gadget graph has Cheeger constant at least 1.   
\end{lemma}
Let $\mathcal{C}$ be a $\llbracket n,k,d \rrbracket$ CSS LDPC code and $\overline{X}_j$ the logical operator to measure.
First, we construct a ``deformed'' code that has $k=1$ logical qubits by attaching measurement gadgets to the $k-1$ logical representatives distinct from $\overline{X}_j$. 
If any logical supports overlap, we construct a deformed code using a brute-force branching of the $Z$ logicals $\overline{Z}_i, i \neq j$ of $\mathcal{C}$ to obtain representatives with disjoint supports. 
Then we attach a measurement gadget from \cite{Williamson2024,Ide2025} to every logical representative $\overline{Z}'_i, i \neq j$ (i.e., either on a leaf of the brute-force branching for initially overlapping logicals, or at the logical support directly if it is disjoint and does not require branching) to obtain the deformed code $\mathcal{C}'$.

By the distance preserving properties of the applied gadgets and in particular the above lemmas, the resulting deformed code has distance $d(\mathcal{D}) \geq d(\mathcal{C})$.
Since each attached measurement gadget has code dimension $1$ \cite{Williamson2024,Ide2025}, the deformed code has dimension $k' = k - (k-1) = 1$. 
Moreover, the deformed code is LDPC since both brute-force branching and the gauging schemes preserve the LDPC property.
Finally, note that we use brute-force branching and the gauging measurement scheme to construct a deformed code, as opposed to conducting the protocol on the input code.

The construction explained above is depicted in the following diagram, where $C_\bullet$ is the chain complex corresponding to the initial code, $C'_\bullet$ the deformed code obtained from brute-force branching, and $F_\bullet$ the complex of the gauging gadget
\begin{equation}\label{eq:parallel-anc-complex-diag}
    \begin{tikzcd}
                                                    & F_1 \arrow[r, "\delta^F"] \arrow[d, "f_1"]        &  F_2 \arrow[d, "f_2"] \arrow[r] & F_3 \\
     \arrow[d, "\gamma_2"] C'_0 \arrow[r,"\delta'_0"]\arrow[d]  & \arrow[d, "\gamma_1"] C'_1 \arrow[r,"\delta'_1"]  & \arrow[d, "\gamma_0"] C'_2 \\
    C_0 \arrow[r,"\delta_0"] & C_1 \arrow[r,"\delta_1"] & C_2
    \end{tikzcd} 
\end{equation}
Hence, the ancilla complex $D_\bullet$ has vector spaces
\begin{align}
   C'_2 \oplus F_3 \to C'_1 \oplus F_2 \to  C'_0 \oplus F_1\,
\end{align}
with boundary maps
\begin{align}
    \partial_2 = \begin{pmatrix}
        \delta'_0   & 0\\
        f_1         & \delta_1^F
    \end{pmatrix}^T, \quad
    \partial_1 = \begin{pmatrix}
        \delta'_1   & 0\\
        f_2         & \delta_2^F
    \end{pmatrix}^T \, ,
\end{align}
obtained by forming the cone complex of the cochain complex.

The chain map from the deformed code to the initial code $\gamma_\bullet \colon D_\bullet \to C_\bullet$ maps elements from $C'$ via the inclusion map to the corresponding elements in $C$, and maps elements from $F_\bullet$ to zero.

Note that it is important that we construct the deformed code by attaching measurement gadgets to the other $k-1$ many $Z$ logicals if we want to measure $\overline{X}_j$, since otherwise we can in general not define $\gamma_\bullet$ in a meaningful manner, which is why the complex in Eq.~(\ref{eq:parallel-anc-complex-diag}) arises from the cochain complex.

As a consequence, we can use $D_\bullet$ as the ancilla complex and apply expansion boosting from Section.~\ref{sec:boost}.
Then Algorithm~\ref{alg:single-shot-surgery} gives use a procedure for measuring any $X$ or $Z$ operator of $\mathcal{C}$.

\section{Case study} \label{sec:case_study}

Abelian multi-cycle codes are a higher dimensional generalization of generalized bicycle codes that have been found to have competitive performance under circuit level noise~\cite{Lin2025, Aasen2025}. 
We focus on the $4$-dimensional case as these are single-shot decodable in both the $X$ and $Z$ sectors~\cite{Lin2025}. 

Let $l \in \mathbb{N}^+$ and $x_l$ be the $l \times l$ shift matrix. 
The $4$-dimensional multi-cycle code is defined using using four circulant matrices $A, B, C, D$, which are each defined by polynomials over $x_l$. 
Then the check matrices $H_X, H_Z$ and meta-check matrices $M_X, M_Z$ are given by
\begin{align}
    M_X &= \begin{pmatrix}
        A^T & B^T & C^T & D^T \\
    \end{pmatrix} \\
    H_X &= \begin{pmatrix}
        B^T & C^T & 0 & D^T & 0 & 0 \\
        A^T & 0 & C^T & 0 & D^T & 0 \\
        0 & A^T & B^T & 0 & 0 & D^T \\
        0 & 0 & 0 & A^T & B^T & C^T
    \end{pmatrix} \\
    H_Z &= \begin{pmatrix}
        C & B & A & 0 & 0 & 0 \\
        D & 0 & 0 & B & A & 0 \\
        0 & D & 0 & C & 0 & A \\
        0 & 0 & D & 0 & C & B
    \end{pmatrix} \\
    M_Z &= \begin{pmatrix}
        D & C & B & A
    \end{pmatrix}\, .
\end{align}
\paragraph{Gadget construction.}
We focus on the $\llbracket 42, 6, 4\rrbracket$ code from Ref.~\cite{Lin2025} given by $l=7$, and the matrices $A = 1 + x_l$, $B = 1 + x_l^2$, $C = 1 + x_l^3$, $D = 1 + x_l^4$. 
Following the notation of Ref.~\cite{Bravyi2024}, we can describe $X$-type logical operators using a polynomial for each column block of physical qubits. We consider the $X$-type logical operator which has the following disjoint representations
\begin{align}
    \overline{X}_\alpha = X\left( \alpha(1+x^4), 0, \alpha x, 0, \alpha x^4, 0 \right) &&\text{for } \alpha \in \left\{1, x, x^2 \right\}.
\end{align}
Then we can construct a sub-complex corresponding to a quantum code that includes these logical operators as follows. 
The $X$-type checks in the subcomplex are given by the minimal set of generators that clean $\overline{X}_1$ to $\overline{X}_x$ and $\overline{X}_x$ to $\overline{X}_{x^2}$. 
These are given by
\begin{align}
    &X\left(\alpha(1+x), 0, \alpha(x+x^5), 0, \alpha(x+x^4)\right) \\
    &X\left(\alpha(x^4+x^5), 0, \alpha(x^2+x^5), 0, \alpha(x+x^5) \right) &&\text{for } \alpha \in \left\{1, x\right\}
\end{align}
The physical qubits in the sub-complex are the set of qubits in the support of the above $X$-type checks. 
These are the qubits with indices
\begin{equation}
    \left\{ 0, 1, 2, 4, 5, 6, 15, 16, 17, 19, 20, 29, 30, 32, 33, 34 \right\}.
\end{equation}
The $Z$-checks in the sub-complex are the set of checks whose support intersect non-trivially with these physical qubits. These are given by

\begin{align}
    Z\left(\alpha(1+x^3), \alpha(1+x^2), \alpha(1+x), 0, 0, 0 \right) \\
    Z\left(\alpha(1+x^4), 0, 0, \alpha(1+x^2), \alpha(1+x), 0 \right) &&\text{for } \alpha \in \left\{1, x, x^2, x^3, x^4, x^5, x^6 \right\}\\
    Z\left(0, 0, \beta(x+x^5), 0, \beta(x+x^4), \beta(x+x^2) \right) &&\text{for } \beta \in \left\{1, x, x^2, x^3, x^4, x^5 \right\}.
\end{align}

This sub-complex then defines the code $D$ that we use to perform the logical measurement. This code encodes one logical qubit and we verify numerically that $d_Z(D) = 3$.
We define the chain map from $D$ to $C$ using the inclusion maps from $i$-chains in $D$ to $i$-chains in $C$. 
We attach this gadget to the base code using only one layer and numerically confirm that the distance of the merged code is $4$, hence the merged code has parameters $\llbracket 62, 5, 4\rrbracket$ and the surgery procedure has a fault distance of $3$.

\paragraph{Numerical study.}
We have conducted decoding simulations of measuring a given logical operator with our scheme and the ancilla complex $D$ constructed from the considered multi-cycle code instance to provide preliminary numerical underpinnings for the proposed scheme.

To demonstrate the logical error suppression and compare our `fast' scheme to the standard logical measurement approach~\cite{Ide2025,Williamson2024}, we construct logical measurement gadgets for both schemes.

For the standard scheme we construct the ancilla complex from the dual Tanner graph of the induced graph on the logical operator to measure and confirm numerically that the gadget is distance preserving. 
The merged code has parameters $\llbracket 48,5,4 \rrbracket$.
For this gadget we simulate $3$ rounds (corresponding to the fault distance) as well as a single round, to analyze the logical error suppression of both. 
For the fast scheme we simulate only a single round. 
We use a phenomenological noise model and decode the $X$ sector only using BP+OSD implemented in the LDPC library~\cite{Roffe_LDPC_Python_tools_2022}.
We use single stage decoding~\cite{Higgott2023}, wherein the meta-check matrices are explicitly inserted into the parity check matrix sent to the decoder.
The logical observables correspond to the unmeasured logical $X$ operators and the product of all checks in the ancilla system, which gives the logical measurement outcome.

The results depicted in Fig.~\ref{fig:numerics} indicate that our scheme with a single round has similar logical error rate as the standard scheme with three rounds, whereas the standard scheme with a single round does not lead to any logical error suppression.
\begin{figure}[t]
    \centering
    \includegraphics[width=0.75\linewidth]{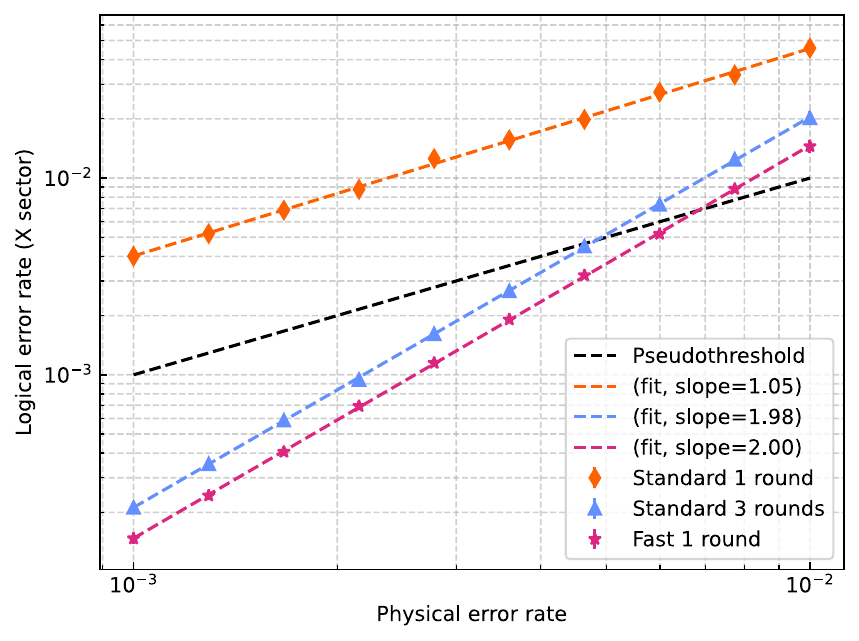}
    \caption{Surgery simulations for a $\llbracket 42,6,4 \rrbracket$ multi-cycle code instance under phenomenological noise for $X$-checks and $X$ logical observables. 
    The line with diamond markers corresponds to the standard scheme with one round and the line with triangles pointing up corresponds to the standard scheme with three rounds. 
    The line with star markers corresponds to the fast surgery scheme with 1 round and the unmarked line illustrates the pseudothreshold ($x=y$).
    Error bars represent $99\%$ binomial confidence intervals and the dashed lines are loglog fits.}
    \label{fig:numerics}
\end{figure}
The numerical data used in this plot is available at~\cite{iceberg_simulation_2024}.

\section{Conclusion}

In this paper we have introduced a framework for performing logical CSS measurements on quantum LDPC codes using fewer rounds of measurements than in previous generalized surgery schemes.
We achieved this by utilizing an auxiliary $2$-dimensional complex and forming the merged code by taking the total complex of the resulting chain map.
The distance of the auxiliary complex ensures the fault-tolerance of the scheme even with a constant number of rounds of measurements, and the systolic expansion ensures the merged code is distance preserving.
We then studied an application of this scheme on Abelian multi-cycle codes, demonstrating how the scheme can be used to practically reduce the time overhead of generalized surgery.

There are several directions to extend upon this work in the future.
The scheme we present in this paper applies to CSS codes and allows us to measure $X$-type and $Z$-type logical Pauli operators.
It is possible to describe a general stabilizer code by using a symplectic chain complex, where the modules are now symplectic vector spaces and the boundary maps must be symplectic.
Can we then describe our surgery procedure in terms of symplectic chain complexes?
This would allow us to measure non-CSS logical operators, such as $Y$-type Pauli operators, which are necessary for implementing the full logical Clifford group without using an external source of logical $\ket{i}$ states.

We give a case study of our scheme on a small example of a multi-cycle code. It would be interesting in the future to explore larger codes, and to tailor the measurement schemes for these codes. This will be necessary to utilize this scheme for large scale fault-tolerant applications. Furthermore, we should design syndrome extraction circuits for fast surgery schemes and benchmark their performance under circuit level noise to better assess their performance on hardware.

\section{Acknowledgements}
We thank Paul Webster, Sam Smith and all our colleagues at Iceberg Quantum for helpful discussions.

\bibliographystyle{apsrev4-1}
\bibliography{refs}

\appendix

\section{Connection to coning}

After some probing of the homological algebra literature \cite{weibel1994introduction}, we realised that the method detailed in the main text is closely related to \emph{coning} -- which was pursued in \cite{Ide2025}.
Let $(B_\bullet, \partial^B), (C_\bullet, \partial^C)$ be chain complexes, and $f_\bullet: B_\bullet \rightarrow C_\bullet$ a chain map.
\begin{equation} 
\begin{tikzcd}[cells={nodes={minimum height=2em}}]
    B_{i+1} \arrow[r,"\partial^B_{i+1}"] \arrow[d, "f_{i+1}"] &  B_i \arrow[r,"\partial^B_{i}"] \arrow[d,"f_i"] &  B_{i-1}\arrow[d,"f_{i-1}"] \\
    C_{i+1} \arrow[r,"\partial^C_{i+1}"] & C_i \arrow[r, "\partial^C_{i}"] & C_{i-1}
\end{tikzcd}
\end{equation}
The reader will note this complex is precisely a double complex, which we relabel to make this observation explicit.
\begin{equation} 
E_{i,j} = \begin{tikzcd}[cells={nodes={minimum height=2em}}]
    E_{i+1,1} \arrow[r,"\partial^B_{i+1}"] \arrow[d, "f_{i+1}"] &  E_{i,1} \arrow[r,"\partial^B_{i}"] \arrow[d,"f_i"] &  E_{i-1,1}\arrow[d,"f_{i-1}"] \\
    E_{i+1,0}  \arrow[r,"\partial^C_{i+1}"] & E_{i,0} \arrow[r, "\partial^C_{i}"] & E_{i-1,0}
\end{tikzcd}
\end{equation}
Then the cone complex $\text{Cone}(f)_\bullet$ \cite{weibel1994introduction}, corresponds exactly to the total complex $\tot (E)_\bullet$:
\[
    \text{Cone}(f)_\bullet = \tot (E)_\bullet
\]
It is well known how to compute the homology of cone complexes.
In particular, we can refer to \cite{weibel1994introduction} 1.5.2, where it is shown that the cone complex induces the following long exact sequence on homology groups:
\begin{equation}
    H_{n+1}(\text{Cone}(f)) \xrightarrow{\pi_\star} H_n(B) \xrightarrow{f_\star} H_n(C) \xrightarrow{\iota_\star} H_{n}(\text{Cone}(f))
\end{equation}
Where $\pi(b, c) = (b,0)$, and $\iota(c) = (c,0)$, and $\pi_\star, \iota_\star, f_\star$ is simply the restriction of $\pi, \iota, f$ to the homology groups. It is always possible to exact a short exact sequence from a long one, which we will leverage here. There are many way to go about this, since, for example $\im \pi = \ker f$, however we want to be careful to not to explicitly refer to $H_{i}(\text{Cone}(f))$, as it is the quantity we are trying to compute. Thus, the only choice we are left with is:
\begin{equation}
    0 \rightarrow \coker f_\star \xrightarrow{\iota_\star} H_{i}(\text{Cone}(f)) \xrightarrow{\pi_\star} \ker f_\star \rightarrow 0
\end{equation}
Where $\coker f_\star = H_i(C) / \im f_\star = H_i(C)/f_\star (H_i(B))$, while we have:
\[
\ker f_\star = \{ b \in H_{i-1}(B) : \exists c_b \in C_i, f(b) = \partial^C_i c_b\}
\] 
We should note, for the sake of completeness, that $\pi^{-1}$ is indeed well defined on as it sends $b$ to $(b,c_b)$. Using the splitting lemma, we get an explicit expression for $H_{i}(\text{Cone}(f))$:
\begin{align}
        H_{i}(\text{Cone}(f)) &= \iota_\star(\coker f_\star) \bigoplus \pi_\star^{-1} (\ker f_\star) \\ 
        &= \label{eq:coning-basis}\iota_\star(H_i(C)/f_\star (H_i(B))) \bigoplus \pi_\star^{-1} (\ker f_\star)
\end{align}
\paragraph{Remark:} Equation \ref{eq:coning-basis} provides an alternative proof of Lemmata \ref{lemma:dimension}, \ref{lem:homology-basis} and \ref{lemma:dz-distance} -- note in our case $\ker f_\star = 0$ because $H_0(B) = 0$ which greatly simplifies calculations.
We present both the original exposition of the fast surgery, and the more natural homological one, to appeal to the broadest audience.

\end{document}